\pdfoutput=1
\documentclass[conference, 10pt]{IEEEtran}
\usepackage{cite}
\usepackage[pdftex]{graphicx}
\DeclareGraphicsExtensions{.JPG,.eps,.pdf}
\usepackage{amsmath}
\usepackage{amsthm}
\usepackage{amsfonts}

\usepackage{algpseudocode}
\usepackage{algorithm}
\usepackage{color}
\usepackage{amscd}
\usepackage{bbm}
\usepackage{stfloats}
\usepackage[normalem]{ulem}
\usepackage{paralist}
\IEEEoverridecommandlockouts

\usepackage{tikz,pgfplots,xcolor}
\usepackage{color}
\usepackage{dsfont}
\usepackage{amssymb}
\newcommand{\be}{\begin{equation}}
\newcommand{\ee}{\end{equation}}
\newcommand{\bea}{\begin{eqnarray}}
\newcommand{\eea}{\end{eqnarray}}
\newcommand{\ben}{\begin{enumerate}}
	\newcommand{\een}{\end{enumerate}}
\newcommand{\beseq}{\begin{subequations}}
	\newcommand{\eeseq}{\end{subequations}}
\algdef{SE}[DOWHILE]{Do}{doWhile}{\algorithmicdo}[1]{\algorithmicwhile\ #1}%

\newlength\figureheight
\newlength\figurewidth
\newtheorem{lem}{Lemma}

\begin{document}

\title{On Distributed Power Control for Uncoordinated Dual Energy Harvesting Links: Performance Bounds and Near-Optimal Policies
\thanks{Mohit K. Sharma and Chandra R. Murthy are with the Dept. of ECE, Indian Institute of Science, Bangalore 560012, India (e-mails: \{mohit, cmurthy\}@ece.iisc.ernet.in).}
		\thanks{Rahul Vaze is with the School of Technology and Computer Science, Tata
	Institute of Fundamental Research, Mumbai 400005, India (e-mail: vaze@tcs.tifr.res.in).}
\thanks{This work was financially supported by a research grant from the Aerospace Network Research Consortium.}}

\author{\IEEEauthorblockN{Mohit K. Sharma, Chandra R. Murthy, \emph{Senior Member, IEEE}, and Rahul Vaze, \emph{Member, IEEE}}}

\maketitle

\begin{abstract}
In this paper, we consider a point-to-point link between an energy harvesting transmitter and receiver, where neither node has the information about the battery state or energy availability at the other node. 
We consider a model where data is successfully delivered only in slots where both nodes are \emph{active}. Energy loss occurs whenever one node turns on while the other node is in sleep mode. In each slot, based on their own energy availability, the transmitter and receiver need to independently decide whether or not to turn on, with the aim of maximizing the long-term \textcolor{black}{time-average} throughput. We  present an upper bound on the throughput achievable by analyzing a genie-aided system that has noncausal knowledge of the energy arrivals \textcolor{black}{at both the nodes}. Next, we propose an online policy requiring an occasional one-bit feedback \textcolor{black}{whose throughput is} within one bit of the upper bound, asymptotically in the battery size.  In order to further reduce the feedback required, we propose a time-dilated version of the online policy. As the time dilation gets large, this policy does not require any feedback and achieves the upper bound \textcolor{black}{asymptotically in the battery size}. Inspired by this, we also propose a near-optimal fully uncoordinated policy. We use Monte Carlo simulations to validate our theoretical results and illustrate the performance of the proposed policies. 

\begin{keywords}
Throughput, energy harvesting, decoding cost, distributed power control.
\end{keywords}
\end{abstract}
\newcounter{MYtempeqncnt}

\section{Introduction}
Fifth generation wireless networks aim to  provide connectivity to massive numbers of low power sensors, deployed to collect the data for monitoring and surveillance. Often, these applications require the sensors to be deployed at places that are not easily accessible. Thus, the lifetime of a sensor is limited by the battery attached to it. The energy harvesting (EH) technology circumvents this problem, as an EH sensor can harvest the energy from the environment, e.g., from a solar source \cite{Ulukus_JSAC_Mar2015, Yuen_IEEEcomMag_Apr2015}. On the other hand, the design of EH communication systems becomes more challenging because of the randomness in the energy arrivals. Specifically, \textcolor{black}{due to the finiteness of the energy storage (battery), it is necessary to design power management policies that strike the right balance between ensuring energy availability when required and avoiding loss of harvested energy when the battery is full.}

The design of power management policies for point-to-point links where both the transmitter and receiver are energy harvesting nodes (EHNs), termed as \emph{dual} EH links, is undertaken in \cite{RVaze_TIT_Aug16, Arafa_JSAC_Dec2015, Doshi_ICCS_Nov2014, Zhou_JSAC_Mar2015, Ayadav_ICC2015} with the objective of optimizing the transmission time \cite{RVaze_TIT_Aug16}, throughput \cite{Arafa_JSAC_Dec2015, Doshi_ICCS_Nov2014}, delay-limited throughput \cite{Zhou_JSAC_Mar2015} and packet drop probability \cite{Ayadav_ICC2015}. However, \textcolor{black}{most of the work on dual EH} considers noncausal knowledge of energy arrivals at both the transmitter and receiver, and the issue of coordination between the transmitter and receiver \textcolor{black}{is neglected for  ease of analysis}. 
Dual EH links constitute a basic building block for wireless sensor networks where all nodes are EHNs. Hence, the design of online policies for energy management in dual EH links is  important.

In this paper, we consider a dual EH link between an EH transmitter and \textcolor{black}{an EH} receiver, communicating over a  Gaussian channel. Neither node has information about the battery state or energy arrivals of the other node. Hence, both the transmitter and the receiver remain uncertain if the other node will be `\emph{on}' in the next slot. Data is delivered successfully if and only if both the transmitter and the receiver are simultaneously `on' in \textcolor{black}{any} slot. 
It is worth noting that perfect coordination between the transmitter and receiver can be achieved if the receiver is allowed to send an acknowledgment for a successfully received packet~\cite{Mohit_TWCJune16_submitted}. However, for a link with \textcolor{black}{an AWGN} channel, \textcolor{black}{since the success or failure of an attempt can be inferred from the rate of transmission}, sending an acknowledgment signal is not necessary and is an additional overhead for energy-starved sensors. The goal of this work is to design a distributed and online power control policy for both the nodes, \textcolor{black}{to} maximize the \textcolor{black}{long-term} time-averaged throughput with  minimal feedback about the battery state at the other node. Our main contributions are as follows:
\begin{enumerate}
	\item First, we derive an upper bound on the maximum achievable throughput, by analyzing a system that has non-causal knowledge of energy arrivals.
	\item Next, we present an online, distributed power control policy \textcolor{black}{whose throughput is within one bit of} the upper bound, and requires an occasional one bit feedback. 
	\item In order to further reduce the amount of feedback to achieve the upper bound, we propose a time-dilated policy \textcolor{black}{which achieves the upper bound and requires zero feedback, in the limit as the time-dilation gets large.}
	\item We also propose a near-optimal, deterministic, fully uncoordinated policy which requires no feedback \textcolor{black}{about the battery state of the other node, and analytically characterize its gap from the optimality.}  
	\item Our simulation results confirm the theoretical findings and illustrate the trade-offs between system parameters. 
\end{enumerate}

The policies presented here not only achieve \textcolor{black}{rates close to} the upper bound but are also simple to implement. \textcolor{black}{Furthermore, the proposed policy requires knowledge about only first and second order statistics of the harvesting processes.} Our policies allow the nodes to operate in a truly uncoordinated fashion, and nearly obviates the need for any feedback.
Also, in contrast to \cite{Doshi_ICCS_Nov2014}, where the upper bound is obtained by assuming unit-sized batteries at the nodes, the upper bound presented here is independent of the battery capacity.  

\section{System Model}
We consider a point-to-point link where an energy harvesting node (EHN) \textcolor{black}{needs to} transmit data to another EHN over a Gaussian channel. The harvesting process at the transmitter and receiver are assumed to be stationary and ergodic random processes with their mean harvesting rates denoted by $\mu_t$ and $\mu_r$, respectively. 
The energy harvested at the transmitter and receiver in the $n^{\text{th}}$ slot is denoted by $\mathcal{E}_t(n)$ and $\mathcal{E}_r(n)$, respectively. \textcolor{black}{At both the nodes, the harvested energy is stored in a perfectly efficient, finite capacity battery}. Since the amount of energy harvested  is random, both the transmitter and  receiver do not know the \textcolor{black}{exact battery state} at their counterpart. 
Hence, in any slot, the transmitter does not know if the receiver will be `on' to receive the data or not, and vice-versa. 
A data packet is successfully delivered if and only if  both the transmitter and receiver are simultaneously \textit{on} in a slot. 

The power control policy at the transmitter and receiver over an $N$ slot horizon is denoted by $\mathcal{P}_t=\{p_t(n)\}_{n=1}^N$ and $\mathcal{P}_r=\{p_r(n)\}_{n=1}^N$, respectively, where $p_t(n)$ and $p_r(n)$ denote the energy used by the transmitter and receiver, respectively, in the $n^{\text{th}}$ slot. The power control policy at the receiver, $p_r(n)$, is binary valued, i.e., if it decides to turn \textit{on} in a slot, it always consumes $R$ units of energy; and it does not incur any energy cost in the sleep mode\cite{MSharma_JSAC_Dec2016, RVaze_TIT_Aug16, Ayadav_ICC2015, Zhou_JSAC_Mar2015}. On the other hand, the transmit power control policy is continuous valued. Without loss of generality, we assume that each slot is of unit duration; hence, we use the terms power and energy interchangeably. By the principle of energy conservation, the battery at the transmitter evolves as 
\begin{equation}
B_{n+1}^t= \min\left\{\max\{0,B_n^t+\mathcal{E}_t(n) -p_t(n)\},B_{\max}^t\right\}.
\label{eq:Tax_battery_evol}
\end{equation}
In the above, $B_n^t$ denotes the battery at the transmitter at the start of the $n^{\text{th}}$ slot, and $B_{\max}^t< \infty$ denotes the size of the battery at the transmitter. The battery at the receiver is of size $B_{\max}^r$, and its state $B_n^r$ evolves in a similar fashion. \textcolor{black}{We also consider the use of  super-capacitor at the transmitter, for temporarily holding energy budgeted for transmission. Its role in the operation of the transmitter will be elaborated on later.}

We assume that the rate achieved corresponding to power $p_t(n)$ is well approximated by the capacity expression, i.e., $\mathcal{R}(p_t(n))\triangleq\log(1+p_t(n))$\textcolor{black}{\cite{Dong_JSAC_Mar2015, Xu_JSAC_Feb2014}, which is an upper bound on the actual rate}.  For simplicity, we assume that the power spectral density of the additive white Gaussian noise at the receiver is unity. Our aim in this paper is to devise a distributed power control strategy for the transmitter and the receiver, i.e., $\mathcal{P}_t$ and $\mathcal{P}_r$, such that the \textcolor{black}{long-term} time-averaged throughput is maximized. That is, our goal is to maximize 
\begin{equation}
\mathcal{T}\triangleq\frac{1}{N}\sum_{n=1}^N\mathbbm{1}_{\{p_r(n)\neq 0\}}\log(1+p_t(n)).
\label{eq:time_average_throughput}
\end{equation}
In the above, $\mathbbm{1}_{\{p_r(n)\neq 0\}}$ is an indicator function which takes value one if $p_r(n)$ is nonzero, otherwise it is equal to zero. Thus, $\mathbbm{1}_{\{p_r(n)\neq 0\}}\log(1+p_t(n))$ denotes the rate achieved in the $n^{\text{th}}$ slot, which is nonzero if and only if  both the EHNs are on in the $n^{\text{th}}$ slot, i.e., $p_t(n)\neq 0$ and $p_r(n)\neq 0$. 

Mathematically, the problem of maximizing the long-term time-averaged throughput can be written as
\begin{subequations}
\begin{align}
\label{eq:Opt_problem_obj}
\qquad\qquad\max_{\{p_t(n), p_r(n),n\geq 1\}} \liminf_{N\to\infty}\mathcal{T}&\\
\hspace{-2cm}\text{subject to: }  0\leq p_t(n)\leq B_t^n,&\\
 p_r(n)\in\{0,R\},  \text{ and }p_r(n)\leq B_r^n.& \label{eq:rxconstr}
\end{align}
\label{eq:opt_problem_form}  
\end{subequations}
In \eqref{eq:rxconstr}, the constraint $p_r(n)\in\{0,R\}$ denotes the fact that the power control policy at the receiver is binary-valued, i.e., it consumes $0$ or $R$ units of energy, depending on whether it is off or on, respectively. \textcolor{black}{Note that, for a given sample path of the
harvesting processes and deterministic policies conditioned on the sample path, $\liminf_{N\to\infty} \mathcal{T}$ is a well defined
deterministic quantity.} We seek to obtain the power control policy for the transmitter and receiver such that they can operate without requiring knowledge of each other's battery state, while achieving near-optimal performance. First, in order to benchmark the performance of any policy, we derive an upper bound on the throughput in~\eqref{eq:opt_problem_form}.

\section{Upper Bound on the Throughput}
\label{Sec:upper_bound}
In this section, we derive an upper bound on the achievable long-term time-averaged throughput by considering a system in which both the EHNs are equipped with infinite size batteries, and have noncausal information about the energy arrivals. The following Lemma provides the upper bounds. 
\begin{lem}
The long-term time-averaged throughput of a dual EH link satisfies:
\begin{enumerate}[a)]
	\item $\liminf_{N\to\infty}\mathcal{T} \leq \log (1+\mu_t)$ \text{ if } $\frac{\mu_r}{R}> 1$,
	\item $\liminf_{N\to\infty}\mathcal{T} \leq \left(\frac{\mu_r}{R}\right)\log \left(1+\frac{R\mu_t}{\mu_r}\right)$ \text{ if } $\frac{\mu_r}{R}\leq 1$.
\end{enumerate}
\label{lem:throughput_upperbound}
\end{lem}
\begin{proof}
See Appendix~\ref{App:proof_upper_bound}.
\end{proof}

In the above Lemma, the first scenario \emph{(a)} corresponds to the setting where the average harvesting rate at the receiver exceeds $R$, the energy consumed by it per slot when it is \textit{on}. Thus, the battery state at the receiver has a positive drift even if it remains on in all slots, i.e., the receiver is energy unconstrained. \textcolor{black}{This case is equivalent to having only the transmitter as an EHN.} Case \emph{(b)} corresponds to a scenario when the receiver is energy-constrained, i.e., the average energy harvested in a slot is less than the energy consumed in one slot. \textcolor{black}{Consequently,} the receiver can only turn on intermittently. To avoid loss of energy, the transmitter must avoid sending data when the receiver is off. However, this requires the transmitter to know the state of the battery at the receiver. In the next section, we present near-optimal policies for both the scenarios. 

  
\section{\textcolor{black}{Asymptotically} Optimal Policies}     
\label{Sec:opt_policy_1_bit_fb}
\textcolor{black}{In the following, we first consider Case \emph{(a)} and present a policy which asymptotically achieves the upper bound given in Lemma 1, and does not require any feedback about the battery state at the receiver.}
\subsection{Energy Unconstrained Receiver, $\frac{\mu_r}{R}> 1$}
\textcolor{black}{First, when the battery state has a positive drift, it is known that the probability that the receiver does not have sufficient energy to turn on decays exponentially with the size of the battery\cite{Mohit_TWCJune16_submitted}. Consequently, with high probability, the receiver can always remain on, making this case equivalent to the scenario where only the transmitter is EH. The optimal policy in this scenario, denoted by $\mathcal{P}^u$, is the same as the one proposed in \cite{rahulsri_IEEEACM_netw_aug2013}, which is as follows:}
\begin{equation}
p^u(n)=\begin{cases}
\mu_t+\delta_t^+, \qquad &B_n^t\geq \frac{B_{max}^t}{2},\\	
\min\{B_n^t, \mu_t-\delta_t^-\}, \qquad &B_n^t< \frac{B_{max}^t}{2},\\
\end{cases}
\label{eq:opt_policy_uncons_scen}
\end{equation}
where  $\delta_t^+ =\delta_t^-=\beta_t\sigma_t^2\frac{\log B_{\max}^t}{B_{\max}^t}$. Here, $\sigma_t^2$ denotes the asymptotic variance of the harvesting process at the transmitter, and $\beta_t\geq 2$ is a constant. It is shown in \cite{rahulsri_IEEEACM_netw_aug2013} that the policy $\mathcal{P}^u\triangleq\{p^u(n)\}_{n=1}^N$ converges to the optimal utility   at the rate $\Theta\left(\left(\frac{\log B_{\max}^t}{B_{\max}^t}\right)^2\right)$ while the transmitter battery discharge probability simultaneously goes to zero at the rate $\Theta\left({B_{\max}^t}^{-\beta_t}\right)$. The transmitter battery discharge probability is defined as $p_d^t\triangleq\lim_{N\to\infty}\frac{1}{N}\sum_{n=1}^N\mathbbm{1}_{\{B_n^t= 0\}}$. 
The receiver battery discharge probability is defined similarly. Thus, in the scenario described in the Case \emph{(a)}, $\mathcal{P}^u$ is  \textcolor{black}{asymptotically} optimal as the battery size gets large. 

\textcolor{black}{Next, we present a policy which achieves the throughput within one bit of the upper bound for Case~\emph{(b)},
with an occasional one bit feedback about the receiver's battery state.}
\subsection{Energy Constrained Receiver, $\frac{\mu_r}{R}<1$}
In this section, we present a policy that requires \emph{occasional} one bit feedback. Qualitatively, the policy operates as follows. The receiver sends a one bit feedback whenever the battery level crosses the half-full mark. We assume that the feedback is received without error and delay, and ignore the energy and time overhead in sending it. The one bit feedback enables the transmitter to track whether the receiver's battery is more than or less than half full. Further, the receiver executes a deterministic policy in either half of the battery state; and the transmitter follows the receiver's policy and transmits only in slots where the receiver is also on. In the slots when the receiver is off, the transmitter accumulates the energy prescribed by its own policy in a super capacitor and uses the accumulated energy for transmission in the next slot when the receiver turns on. \textcolor{black}{The consequence of using a super capacitor to temporarily store energy is that the battery energy discharge at the transmitter depends only on its own battery state; specifically, it is independent of the policy at the receiver.}

We now describe the policy in mathematical terms. The energy accumulated in the super capacitor by the end of $n^{\text{th}}$ slot, given that the transmitter does \emph{not} transmit in that slot, is given by 
\be
\mathcal{C}_e(n) =\begin{cases}
	\mathcal{C}_e(n-1)+\mu_t+\delta_t^+, & \text{if }  B_{n}^t\geq \frac{B_{\max}^t}{2}, \\
		\mathcal{C}_e(n-1)+\min{\{\mu_t-\delta_t^-,B_n^t\}}, & \text{if }  B_{n}^t< \frac{B_{\max}^t}{2},
\end{cases}
\ee
and $\mathcal{C}_e(n)=0$, if data is transmitted in the $n^{\text{th}}$ slot, and $\delta_t^+=\delta_t^-=\beta_t\sigma_t^2\frac{\log B_{\max}^t}{B_{\max}^t}$. Here, we assume that the capacity of the super capacitor is sufficient to store the energy accumulated between two consecutive data transmissions. Let $\mathbbm{1}_{\mathcal{R}^+}$ denote an indicator function which takes the value one if $B_n^r \ge \frac{B_{\max}^r}{2}$ and zero otherwise. Also, let $N_r^+\triangleq\lfloor\frac{R}{\mu_r}\rfloor$, and $N_r^-\triangleq\lceil\frac{R}{\mu_r}\rceil$. In the $n^{\text{th}}$ slot, the transmitter follows the policy $\mathcal{P}_t^c$ given by
\be
p^c_t(n)=\begin{cases}
	\mathcal{C}_e(n-1)+\mu_t+\delta_t^+, & \hspace{-0.5in}\text{if } B_n^t\geq \frac{B_{\max}^t}{2},\\&\hspace{-1.4in} n= N_{\text{on}}+N_r^+\mathbbm{1}_{\mathcal{R}^+}+N_r^-(1-\mathbbm{1}_{\mathcal{R}^+}),\\
	\mathcal{C}_e(n-1)+\min{\{\mu_t-\delta_t^-,B_n^t\}},& B_n^t< \frac{B_{\max}^t}{2}, \\&\hspace{-1.4in} n= N_{\text{on}}+N_r^+\mathbbm{1}_{\mathcal{R}^+}+N_r^-(1-\mathbbm{1}_{\mathcal{R}^+}),\\
	0 &\text{otherwise}.
\end{cases}
\label{eq:opt_Policy_cons_tx}
\ee
In the above, $N_{\text{on}}$ denotes the previous slot when the transmitter and receiver were scheduled to turn on. It is initialized to zero at the first slot ($N_{\text{on}} = 0$ when $n=0$), and at any slot index $n$ satisfying $n = N_{\text{on}}+N_r^+\mathbbm{1}_{\mathcal{R}^+}+N_r^-(1-\mathbbm{1}_{\mathcal{R}^+})$, the transmitter and receiver make an attempt if they have energy, and $N_{\text{on}}$ is set to $N_{\text{on}} = n$, i.e., it is updated to the current slot index. 
The policy \eqref{eq:opt_Policy_cons_tx} is derived using the policy given in \eqref{eq:opt_policy_uncons_scen}, i.e., in each slot, the transmitter computes the energy prescribed by $p^u(n)$ for that slot, and transfers the energy from the battery to the super capacitor. In a slot when the receiver is on, the transmitter uses all the energy accumulated in the super capacitor till that slot to transmit its data.

The policy at the receiver is given as
\be
p^c_r(n)=\begin{cases}
	R,\quad & B_n^r\geq \frac{B_{\max}^r}{2}, n= N_{\text{on}}+N_r^+\\
	R,& R\leq B_n^r<\frac{B_{\max}^r}{2},  n= N_{\text{on}}+N_r^- ,\\
	0 & \text{otherwise}. 
\end{cases}
\label{eq:opt_Policy_cons_rx}
\ee
The receiver's policy $\mathcal{P}^c_r\triangleq\{p_r^c(n)\}_{n=1}^N$ also emulates the policy  $\mathcal{P}^u$ given in \eqref{eq:opt_policy_uncons_scen}. Specifically, the receiver executes a policy similar to $\mathcal{P}^u$ by turning on after $N_r^+\triangleq\lfloor\frac{R}{\mu_r}\rfloor$ slots (resulting in a negative drift in the battery state) if battery is more than half full, otherwise it turns on after $N_r^-\triangleq\lceil\frac{R}{\mu_r}\rceil$ slots (resulting in a positive drift in the battery state). 



In the discussion to follow, let $\mathcal{P}^c$ denote the joint power management policy proposed above, i.e., $\mathcal{P}^c_t\triangleq\{p_t^c(n)\}_{n=1}^N$ and $\mathcal{P}^c_r$ given by \eqref{eq:opt_Policy_cons_tx} and \eqref{eq:opt_Policy_cons_rx}, respectively. The following Lemma asserts that the throughput achieved by the policy $\mathcal{P}^c$ is within 1 bit of the upper bound, when the battery capacity is large. In addition, the probability of battery discharge decays polynomially with the battery size at the transmitter, and it decays exponentially fast with the battery size at the receiver.

\begin{lem} \textcolor{black}{Let $\mathcal{T}^c$ denote the time-average throughput achieved by the policy $\mathcal{P}^c$. For policy $\mathcal{P}^c$, the battery discharge probability at the transmitter and receiver are $p_d^t=\Theta\left({B_{\max}^t}^{-\beta_t}\right)$ and $p_d^r=\Theta\left(\exp\left(-\frac{B_{\max}^r\mu_r\delta_r^-}{\sigma_r^2}\right)\right)$, respectively, where $\beta_t\geq 2$ and $\delta_r^-\triangleq N_r^--N_r$, with $N_r\triangleq\frac{R}{\mu_r}$. In addition,  $\left(\frac{\mu_r}{R}\right) \log \left(1+\frac{R\mu_t}{\mu_r}\right)-\mathcal{T}^c-1 = O\left(\frac{\log B_{\max}^t}{B_{\max}^t}\right).$}	\label{lem:policy_perf_guarantee_uncons}
\end{lem}
\begin{proof}
	See Appendix~\ref{Sec:proof_policy_perf_guarantee_uncons}.
\end{proof}

A careful examination of the proof of Lemma~\ref{lem:policy_perf_guarantee_uncons} reveals that the one bit gap in the throughput arises because of the receiver's policy. From \eqref{eq:opt_Policy_cons_rx}, the receiver's policy is to wake up once in $N_r^+$ slots if its battery is more than half full, and to wake up once in $N_r^-$ slots if its battery is less than half full.  Due to this, the drift in the receiver's battery remains fixed at $\delta_r^- = N_r^- - N_r$ when $B_n^r < B_{\max}^r/2$ and $\delta_r^+ \triangleq N_r - N_r^+$ when $B_n^r \ge B_{\max}^r/2$, irrespective of the value of $B_{\max}^r$. In order to close the gap, we need finer control over the battery drift at the receiver. We need it to be of the order \textcolor{black}{$o( 1/B_{\max}^r)$}, similar to that at the transmitter. This can be achieved using  time-dilation, as described next. 

\section{Optimal Throughput via Time-dilation} \label{sec:time_dilate}

The key idea behind time dilation is to spread the drift $\delta_r^+$ and $\delta_r^-$ \textcolor{black}{at the receiver} over a larger number of slots, resulting in a smaller per-slot drift. That is, instead of \eqref{eq:opt_Policy_cons_rx}, which operates in batches of $\lfloor\frac{R}{\mu_r}\rfloor$ or $\lceil\frac{R}{\mu_r}\rceil$ slots, we consider a policy that operates in batches of $\lfloor\frac{Rf(B_{\max}^r)}{\mu_r}\rfloor$ and $\lceil\frac{Rf(B_{\max}^r)}{\mu_r}\rceil$ slots, where $f(\cdot) > 1$ is a time-dilation function. For example, if $f(B_{\max}^r)$ is an integer, the time dilated policy turns the receiver on for $f(B_{\max}^r)$ slots out of $\lfloor\frac{Rf(B_{\max}^r)}{\mu_r}\rfloor$ slots if the battery at the receiver is more than half full, and it turns the receiver on for $f(B_{\max}^r)$ slots out of $\lceil\frac{Rf(B_{\max}^r)}{\mu_r}\rceil$ if the battery is less than half full. This results in a drift of
\begin{equation}
\delta_{r,f}^+(B_{\max}^r)=f(B_{\max}^r)N_r-\left\lfloor\frac{Rf(B_{\max}^r)}{\mu_r}\right\rfloor
\end{equation}
\begin{equation}
\delta_{r,f}^-(B_{\max}^r)=\left\lceil\frac{Rf(B_{\max}^r)}{\mu_r}\right\rceil-f(B_{\max}^r)N_r
\end{equation}
over $\lfloor\frac{Rf(B_{\max}^r)}{\mu_r}\rfloor$ and $\lceil\frac{Rf(B_{\max}^r)}{\mu_r}\rceil$ slots, respectively. Hence, the per-slot drift is given by 
\be 
\delta_{\text{eff}}^+=\frac{\delta_{r,f}^+(B_{\max}^r)}{\lfloor\frac{Rf(B_{\max}^r)}{\mu_r}\rfloor} \text{\hspace{5pt} and \hspace{5pt}} \delta_{\text{eff}}^-=\frac{\delta_{r,f}^+(B_{\max}^r)}{\lceil\frac{Rf(B_{\max}^r)}{\mu_r}\rceil}.
\ee

\textcolor{black}{The transmit policy is still determined according to \eqref{eq:opt_Policy_cons_tx}. Furthermore, with the help of the one bit feedback, the transmitter can ensure that it transmits only in the $f(B_{\max}^r)$ slots when the receiver is `on'.} It can be shown that the dynamics of the policy under time dilation is similar to the dynamics of the policy $\mathcal{P}^c$, as long as the dilation function $f(B_{\max}^r)$ is sub-linear in $B_{\max}^r$. As a consequence, the proof of Lemma~\ref{lem:policy_perf_guarantee_uncons} can be extended to the time dilated policy as well, and the gap from the upper bound in this case can be made to approach zero as the battery size at the receiver gets large. Also, the policy operates over a longer time-window, has a smaller per-slot drift \textcolor{black}{and lower rate of crossing the half-full mark, resulting in a smaller feedback overhead.} We omit the details here due to lack of space. 


In the next section, we propose a near-optimal policy which operates without any feedback from the receiver, and yet achieves a throughput close to the policy~$\mathcal{P}^c$.

\section{A Policy for Fully Uncoordinated Links}
In this section, we propose an uncoordinated policy which prescribes a deterministic pattern for the receiver to turn on, and does not require any feedback from the receiver. At the transmitter, the  policy  $\mathcal{P}^{uc}$ follows the same strategy as $\mathcal{P}^c_t$ given by \eqref{eq:opt_Policy_cons_tx}. However, the indicator variable $\mathbbm{1}_{\mathcal{R}^+}$ is not available at the transmitter. Hence, it keeps the \emph{frequency} with which it transmits after $N_r^+$ and $N_r^-$ slots the same as for policy $\mathcal{P}^c$, but executes it in a deterministic pattern. The receiver also turns on in the same deterministic pattern, provided it has the energy to do so. To derive the deterministic pattern according to which the receiver turns on for the policy $\mathcal{P}^{uc}$, we first compute the empirical distribution of the battery states at the receiver in which it turns on after $N_r^+$ \textcolor{black}{slots}, denoted as $\pi_r^+$, under the policy $\mathcal{P}^c$. Then, starting from the first slot, under the policy $\mathcal{P}^{uc}$, the receiver turns on  after $N_r^+$ and $N_r^-$ slots in the same ratio as the policy $\mathcal{P}^c$. That is:
%
	\begin{itemize}
		\item We compute $\frac{n^+}{n^-}=\frac{\sum_{n=1}^N\mathbbm{1}_{\{B_n^t\geq B_{\max}^r\}}}{\sum_{n=1}^N \mathbbm{1}_{\{B_n^t< B_{\max}^r\}}}$, for policy $\mathcal{P}^c$.
		\item The receiver turns on at the last slot of every batch of $N_r^+$ slots for $n^+$ consecutive batches, after that it turns on at the last slot of every batch of $N_r^-$ slots for $n^-$ consecutive batches, and so on.
	\end{itemize}

Note that, in the above, $n^+$ and $n^-$ are integers, which can result in an approximation of the stationary probabilities with which the receiver turns on after $N_r^+$ and $N_r^-$ slots. Using larger integers results in a smaller approximation error, leading to the same empirical distribution in the battery states as for policy $\mathcal{P}^c$. This, in turn, results in the two policies attaining roughly the same average throughput. On the other hand, if $n^+$ and $n^-$ are large, the receiver is essentially executing a policy with a negative and positive drift (respectively) for a large number of consecutive slots, which could increase the battery \textcolor{black}{discharge/overflow} probability, leading to a loss of throughput.


The following Lemma characterizes the difference between the throughput achieved by the policy $\mathcal{P}^c$ and $\mathcal{P}^{uc}$, in terms of battery discharge probability of policy $\mathcal{P}^{uc}$.  
\begin{lem}
	The throughput achieved by the policy $\mathcal{P}^{uc}$, denoted by $\mathcal{T}^{uc}$, satisfies
\begin{equation} 
\mathcal{T}^{c}-\mathcal{T}^{uc}=O(\pi_0^{uc}),
\end{equation}
where $\pi_0^{uc}$ denotes the stationary probability that battery at the transmitter or receiver (or both) is empty, while operating under policy $\mathcal{P}^{uc}$.
	\label{lem:diif_policy_uc_policy_c}
\end{lem} 
\begin{proof}
	See Appendix~\ref{app:proof_lemma_diff_policy_uc_vs_c}.
\end{proof}
The significance of the above result comes from the fact that the battery discharge probability, $\pi_0^{uc}$, can be made to decrease rapidly with the battery size, for a well designed policy. Due to this, the gap between the throughput achieved by $\mathcal{P}^{uc}$ and $\mathcal{P}^c$ can be made small.


\section{Simulation Results}
\label{sec:sim}
We evaluate the performance of the proposed policies by evaluating the time-averaged throughput using Monte Carlo simulations of the system over $10^7$ slots. The harvesting processes at the transmitter and receiver are assumed to be spatially and temporally independent and identically distributed according to the Bernoulli distribution with harvesting probabilities $\rho_t$ and $\rho_r$, respectively. 

Fig.~\ref{fig:energy_uncons} shows the average per slot throughput  when the receiver is energy unconstrained. We note that the policy given in \eqref{eq:opt_policy_uncons_scen} achieves the upper bound derived in Lemma~\ref{lem:throughput_upperbound}. In this case, the harvesting rate at the transmitter completely determines the average throughput performance. 

In Fig.~\ref{fig:energy_cons}, we show the average per slot throughput when the receiver is energy constrained. The performance of policy $\mathcal{P}^c$ given in \eqref{eq:opt_Policy_cons_tx} and \eqref{eq:opt_Policy_cons_rx}, which requires an occasional one bit feedback, is benchmarked against the upper bound. 
We see that the throughput of $\mathcal{P}^c$ is very close to the upper bound. The figure also shows the the time-dilated policy discussed in Sec.~\ref{sec:time_dilate} further closes the gap to the upper bound. In Fig.~\ref{fig:energy_cons_battry_size}, we study the impact of the battery size at the two nodes on the performance of the policy $\mathcal{P}^c$ for a system with an energy constrained receiver. The per slot throughput achieved by the policy is near-optimal even with small capacity batteries.

Finally, in Fig.~\ref{fig:energy_cons_fullyuncoord}, we compare the performance of the policy $\mathcal{P}^{uc}$ against the throughput of the policy $\mathcal{P}^c$. We note that the  throughput achieved by $\mathcal{P}^{uc}$ is only marginally lower than that achieved by $\mathcal{P}^c$. Thus, the price paid for fully uncoordinated operation is quite small.
\begin{figure}[t!]
\begin{center}
\includegraphics[width=3.3in]{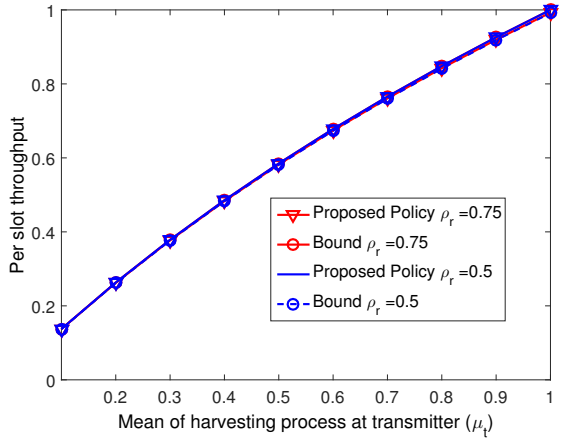}
\caption{Energy unconstrained receiver: The policy presented in \eqref{eq:opt_policy_uncons_scen} achieves the bound. Parameters chosen are $R=0.5$ and $B_{\max}^t=B_{\max}^r=50$. }
\label{fig:energy_uncons}
\end{center}
\end{figure}
\begin{figure}[t!]
	\begin{center}		\includegraphics[width=3.3in]{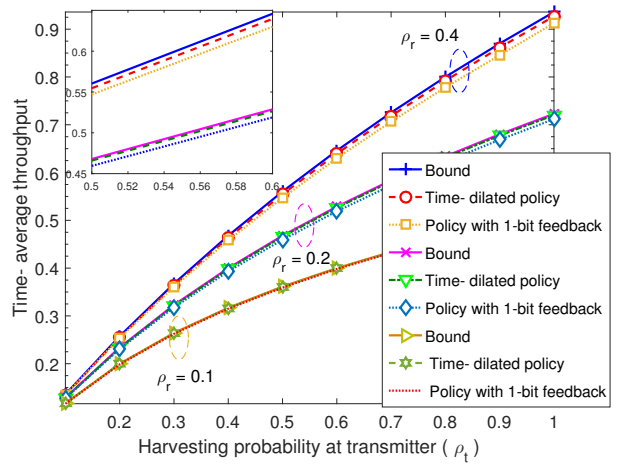}
\caption{Energy constrained receiver: The policy $\mathcal{P}^c$ with occasional one bit feedback,  
achieves a throughput close to upper bound. The time-dilation further improves its performance. The result corresponds to time-dilation $f(\cdot)=100$. Other parameters are $R=0.5$ and $B_{\max}^t=B_{\max}^r=1000$. }
		\label{fig:energy_cons}
	\end{center}
\end{figure}
\begin{figure}[t!]
\begin{center}
\includegraphics[width=3.3in]{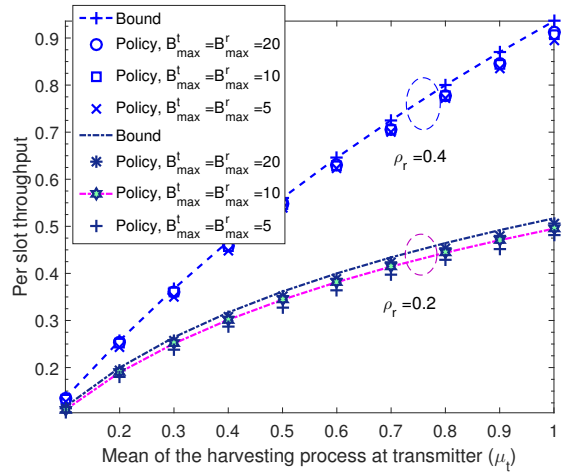}
\caption{Impact of battery size on the throughput of policy $\mathcal{P}^c$, for $R=0.5$. }
\label{fig:energy_cons_battry_size}
\end{center}
\end{figure}
\begin{figure}[h!]
	\begin{center}		\includegraphics[width=3.3in]{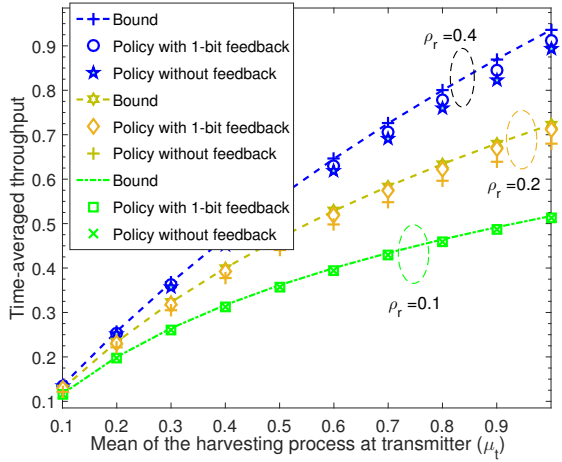}
		\caption{The fully uncoordinated policy $\mathcal{P}^{uc}$ achieves a throughput close to that of the policy $\mathcal{P}^c$. For $\mathcal{P}^{uc}$, the values of $(n^+,n^-)$ are $(5,1)$, $(1,1)$ and $(2,1)$ for $\rho_r = 0.1, 0.2$ and $0.4$, respectively.  Other parameters:  $B_{\max}^t=B_{\max}^r=50$,  $R=0.5$.  }		\label{fig:energy_cons_fullyuncoord}
	\end{center}
\end{figure}
\section{Conclusions}
\label{sec:conclude}
In this paper, we considered the problem of designing power control policies for uncoordinated dual EH links, where both the transmitter and receiver are unaware of the battery state of their counterparts. First, we derived an upper bound on the achievable throughput with the help of a genie-aided system, which has noncausal knowledge of the energy arrivals. Then, we considered a scenario where the receiver is energy unconstrained, and presented a policy which achieves the upper bound. Next, we considered the case of an energy constrained receiver, and presented a policy which achieves the upper bound asymptotically through time-dilation and requires occasional one bit feedback . We also presented a fully uncoordinated policy in which the nodes deterministically make their data transmission attempts, and empirically showed that it achieves near-optimal throughput \textcolor{black}{without requiring any feedback}. Future work could extend the proposed policies to fading channels and spatially correlated harvesting processes. 

\appendix
\subsection{Proof of Lemma~\ref{lem:throughput_upperbound}}
\label{App:proof_upper_bound}
\begin{proof}[Case a)]
From \eqref{eq:time_average_throughput}, the time-averaged throughput of a dual EH link can be upper bounded as 
\begin{align*}
\mathcal{T} &=\frac{1}{N}\sum_{n=1}^{N}\mathbbm{1}_{\{p_r(n)\neq 0\}}\log(1+p_t(n)),\\
&\leq\frac{1}{N}\sum_{n=1}^{N}\log(1+p_t(n)) {\leq} \log\left(1+\frac{1}{N}\sum_{n=1}^{N}p_t(n)\right).
\end{align*} 
The last inequality above follows from Jensen's inequality. Next, taking the limit $N \rightarrow \infty$, we get
\begin{align*}
\liminf_{N\to\infty}\mathcal{T}&\leq\liminf_{N\to\infty}\log\left(1+\frac{1}{N}\sum_{n=1}^{N}p_t(n)\right),\\
&\overset{(c)}{=}\log\left(1+\liminf_{N\to\infty}\frac{1}{N}\sum_{n=1}^{N}p_t(n)\right),\\
&\overset{(d)}{\leq}\log\left(1+\liminf_{N\to\infty}\frac{1}{N}\sum_{n=1}^{N}\mathcal{E}_t(n)\right),\\
&\overset{(e)}{=}\log(1+\mu_t),
\end{align*}
where (c) follows because the logarithm is a continuous function, (d) follows from the \textcolor{black}{fact that the total energy consumed can not exceed the total energy harvested}, and (e) follows from the ergodicity of the harvesting process. This completes the proof in Case (a). 
\end{proof}
\begin{proof}[Case b)]
In this scenario, the receiver can only turn on intermittently, and the lack of information about the battery state of the other node can lead to energy loss at the nodes. To obtain an upper bound, we consider a genie-aided system where the transmitter and receiver are equipped with infinite sized batteries, and \textcolor{black}{the entire energy harvested over $N$ slots is made available in the first slot itself, at both the nodes. In this case}, there is no energy loss due to lack of coordination, as both the transmitter and receiver know the number of slots when the receiver can turn on. Hence, the throughput of this system is an upper bound on \eqref{eq:Opt_problem_obj}. 

From the strong law of large numbers, for large $N$, the energy at the transmitter and the receiver at the beginning of communication is $N\mu_t$ and $N\mu_r$, respectively. Thus, the total number of slots the receiver can remain on is $N'=\lfloor\frac{N\mu_r}{R}\rfloor$. The long-term time-averaged throughput, $\mathcal{T}_g$, of this genie-aided system is
\begin{align}
\liminf_{N\to\infty}\mathcal{T}_g & \leq\liminf_{N\to\infty}\frac{1}{N}\sum_{n=1}^{N'}\log(1+p_t(n)), \nonumber \\
&\leq \liminf_{N\to\infty}\frac{N'}{N}\log\left(1+\frac{N\mu_t}{N'}\right).
\end{align}
The last inequality above is based on the fact that it is optimal to equally allocate the energy available over the $N'$ slots, since the logarithm is a concave function. Noting that $\lim_{N\rightarrow\infty} \frac{N'}{N} = \frac{\mu_r}{R}$ completes the proof.
\end{proof}

\subsection{Proof of Lemma~\ref{lem:policy_perf_guarantee_uncons}}
\label{Sec:proof_policy_perf_guarantee_uncons}
\begin{proof}
	The proof of this Lemma is adapted from \cite{rahulsri_IEEEACM_netw_aug2013}. At the transmitter, we choose $\delta_t^+=\delta_t^-=\beta_t\sigma_t^2\frac{\log B_{\max}^t}{B_{\max}^t}$. On the other hand, at the receiver, $\delta_r^+=N_r-\lfloor\frac{R}{\mu_r}\rfloor$, and $\delta_r^-=\lceil\frac{R}{\mu_r}\rceil-N_r$, where $N_r=\frac{R}{\mu_r}$. First, we analyze the battery discharge probability at the transmitter and  receiver, denoted by $p_d^t$ and $p_d^r$, respectively. We use  \cite[Lem. 2]{rahulsri_IEEEACM_netw_aug2013}, which was derived for the case where only the transmitter is an EHN. 
	
Recall that the energy transferred from the battery to the super capacitor at the transmitter depends  on the battery state at the transmitter, while the feedback sent by the receiver only determines the slot in which the data is transmitted, using the energy accumulated in the super capacitor. \textcolor{black}{Similarly, the decision to turn on at the receiver depends only on the state of the battery at the receiver. Thus, the  batteries at both the transmitter and receiver evolve independently of each other. Hence, the result in \cite[Lem. 2]{rahulsri_IEEEACM_netw_aug2013} is  applicable to our case where both the transmitter and receiver are EHNs. Thus, the battery discharge probabilities at the transmitter and receiver decay as $\Theta\left(\exp\left(-\frac{B_{\max}^t\delta_t^-}{\sigma_t^2}\right)\right)$ and $\Theta\left(\exp\left(-\frac{B_{\max}^r\mu_r\delta_r^-}{\sigma_r^2}\right)\right)$, respectively. Since $\delta_t^+=\delta_t^-=\beta_t\sigma_t^2\frac{\log B_{\max}^t}{B_{\max}^t}$, $p_d^t=\Theta\left({B_{\max}^t}^{-\beta_t}\right), \beta_t > 0$.} Similar results hold  for the battery overflow probabilities also.  
	
	Next, to show that the policy $\mathcal{P}^c$ asymptotically achieves within one bit of the upper bound, we first characterize the rate obtained in a slot. Under policy $\mathcal{P}^c$, the receiver turns on after $N_r^+$ and $N_r^-$ slots, depending on the battery state at the receiver. Thus, the rate obtained in a slot ranges between $\mathcal{R}_{\max}^{N_r^-}\triangleq \mathcal{R}\left(N_r^-(\mu_t+\delta_t^+)\right)$  and $\mathcal{R}_{\min}^{N_r^+}\triangleq\mathcal{R}\left(N_r^+(\mu_t-\delta_t^-)\right)$. Here, $\mathcal{R}_{\max}^{N_r^-}$ denotes the maximum rate obtained in a slot, which is achieved when receiver turns on after $N_r^-$ slots and the battery at the transmitter remains more than half full during all the $N_r^-$ slots. Similarly, $\mathcal{R}_{\min}^{N_r^+}$ is the minimum rate obtained in a slot, which is achieved when the receiver turns on after $N_r^+$ slots and the battery at the transmitter is less than half full for the entire duration of $N_r^+$ slots.
	Since $\mathcal{R}$ is an analytic function, using Taylor's expansion, we can write 
	\begin{multline}
	\mathcal{R}_{\max}^{N_r^-}=\mathcal{R}\left[(N_r+\delta_r^-)(\mu_t+\delta_t^+)\right]\\
	=\mathcal{R}(N_r\mu_t)+\mathcal{R}^{(1)}(N_r\mu_t)\delta_{\max}+\mathcal{R}^{(2)}(N_r\mu_t)\delta_{\max}^2+o(\delta_{\max}^2)\nonumber
	\end{multline}
	where  $\delta_{\max}\triangleq\delta_r^-(\mu_t+\delta_t^+)+N_r\delta_t^+$. Similarly, 
	\begin{align}
	\mathcal{R}_{\min}^{N_r^+}&=\mathcal{R}(N_r\mu_t)+\mathcal{R}^{(1)}(N_r\mu_t)\delta_{\min}+\mathcal{R}^{(2)}(N_r\mu_t)\delta_{\min}^2\nonumber\\
&\qquad\qquad\qquad+o(\delta_{\min}^2) \nonumber
	\end{align}
	where $\delta_{\min}\triangleq\delta_r^+\delta_t^- -N_r\delta_t^- -\mu_t\delta_r^+$.
	 Now, the actual rate achieved depends by the amount of energy used for transmission, which, in turn, depends on the number of slots since the previous transmission attempt. It also depends on the sequence of states the batteries at the two nodes go through, starting from the slot the transmitter previously made an attempt. Hence, the transmit power corresponding to an arbitrary state sequence $s$ can be written as
	 \be
	 p_s=\begin{cases}
	 	N_r^-\mu_t+k_s\delta_t^+-(N_r^--k_s)\delta_t^-,  & \text{ if } s\in\mathcal{S}_{N_r^-},\\
	 	N_r^+\mu_t+\ell_s\delta_t^+-(N_r^--\ell_s)\delta_t^-,  & \text{ if } s\in\mathcal{S}_{N_r^+},
	 \end{cases}
	 \ee
	 where $0\leq k_s\leq N_r^-$ and $0\leq \ell_s\leq N_r^-$ denote the number of slots when the battery at the transmitter is more than half full, when the communication happens in $N_r^-$ and $N_r^+$ slots, respectively. Also, $\mathcal{S}_{N_r^-}$ and $\mathcal{S}_{N_r^+}$ denote the set of sequence of states in which the receiver turns on after  $N_r^-$ and $N_r^+$ slots, respectively. 
	 
	 The total number of bits transmitted corresponding to an arbitrary state sequence $s$ in which the transmit energy is $N_r\mu_t+\delta_s$, is written as 
	 \begin{equation}
	 \mathcal{R}_{s}=\mathcal{R}(N_r\mu_t)+\mathcal{R}^{(1)}(N_r\mu_t)\delta_s+\mathcal{R}^{(2)}(N_r\mu_t)\delta_s^2+o(\delta_s)^2.
	 \label{eq:Us_tylor_expan}
	 \end{equation}
	 In the above, since $N_r^-=N_r+\delta_r^-$ $N_r^+=N_r-\delta_r^+$ and $\delta_t^+=\delta_t^-$, $\delta_s$ (by comparing $p_s$ with $N_r\mu_t+\delta_s$) is given as
	 
	 \be
	 \delta_s=\begin{cases}
	 	\mu_t\delta_r^--N_r^-\delta_t^-+2k_s\delta_t  & \text{ if } s\in\mathcal{S}_{N_r^-}\\
	 	-N_r^+\delta_t^--\mu_t\delta_r^++2\ell_s\delta_t  & \text{ if } s\in\mathcal{S}_{N_r^+}.
	 \end{cases}
	 \ee

	The rates obtained for policy $\mathcal{P}^c$ can also be characterized in terms of  the Markov chain described in the following. In terms of Markov reward process, the rate $\mathcal{R}_s$ can be viewed as the reward obtained when  \textcolor{black}{the Markov chain $\mathcal{M}$} visits the state $s$. The state of Markov chain is given by a set of tuples of battery states at the transmitter and the receiver. Depending on the length of the sequence of tuples of the battery states, the state space of the Markov chain can be partitioned into two disjoint subsets, containing $N_r^+$ and $N_r^-$ length sequences of battery states, denoted by $\mathcal{S}_{N_r^+}$ and $\mathcal{S}_{N_r^-}$, respectively. A typical state $s\in \mathcal{S}_{N_r^+}$ is denoted as $\{(B_m^t,B_m^r)\}_{m=1}^{N_r^+}$. The transition probabilities of this Markov chain can be written in terms of the transition probabilities of the Markov chains describing the evolution of the battery at the transmitter and receiver, given by \eqref{eq:Tax_battery_evol}.  For instance, in a scenario where the transmitter and the receiver harvest the energy according to a Bernoulli process\cite{MSharma_JSAC_Dec2016}, the probability of making a transition from an arbitrary state $s\in \mathcal{S}_{N_r^+}$ to a state $s'\in \mathcal{S}_{N_r^+}$,  in which the reward obtained is $\mathcal{R}_{\min}^{N_r^+}$, can be written as follows. The probability of transition from $s$ to $s'$ is one, if the battery at the transmitter and receiver in the last tuple of the state $s$ is such that $B_{N_r^+}^t< \frac{B_{\max}^r}{2}-N_r^+$ and $B_{N_r^+}^r\geq \frac{B_{\max}^r}{2}+R$,  and for state $s'$, $B_m^t< \frac{B_{\max}^t}{2}$ as well as $B_m^r>\frac{B_{\max}^r}{2}$ for all $1\leq m\leq N_r^+$; otherwise it is zero. Note that, the rate $\mathcal{R}_{\min}^{N_r^+}$ is the reward corresponding to the state which is given by the set of tuples of battery states in which the battery at the transmitter is always less than half full while the battery at the receiver is more than half full.

\textcolor{black}{Under the above Markov chain formulation,} the time-averaged throughput is
\begin{equation}
\mathcal{T}^c=\sum_{s\in\mathcal{S}}\pi_s\frac{\mathcal{R}_s}{N_s},
\label{eq:time_average_utility}
\end{equation} 
where $\pi_s$ denotes the steady state probability of the system being in a state $s$ such that the rate $\mathcal{R}_s$ is obtained in $N$ slots. Note that, the existence of the steady state distribution is ensured by the fact that the Markov chain $\mathcal{M}$ has a finite number of states. Also, in the above, $N_s$ takes the value $N_r^+$ and $N_r^-$ depending on the state s. Next, using \eqref{eq:Us_tylor_expan}, time-averaged throughput in \eqref{eq:time_average_utility} can be rewritten as  
\begin{align}
\mathcal{T}^c&=\sum_{s\in\mathcal{S}_{N_r^-}}\pi_s\frac{\mathcal{R}_s}{N_r+\delta_r^+}+\sum_{s\in\mathcal{S}_{N_r^+}}\pi_s\frac{\mathcal{R}_s}{N_r-\delta_r^+},\\
&=\mathcal{R}(N_r\mu_t)\left[\frac{1}{N_r+\delta_r^+}\sum_{s\in\mathcal{S}_{N_r^-}}\pi_s+\frac{1}{N_r-\delta_r^+}\sum_{s\in\mathcal{S}_{N_r^+}}\pi_s\right]\nonumber\\
&+\mathcal{R}^{(1)}(N_r\mu_t)\left[\frac{1}{N_r^-}\sum_{s\in\mathcal{S}_{N_r^-}}\pi_s\delta_s+\frac{1}{N_r^+}\sum_{s\in\mathcal{S}_{N_r^+}}\pi_s\delta_s\right]\nonumber\\
&+\mathcal{R}^{(2)}(N_r\mu_t)\left[\frac{1}{N_r^-}\sum_{s\in\mathcal{S}_{N_r^-}}\pi_s\delta_s^2+\frac{1}{N_r^+}\sum_{s\in\mathcal{S}_{N_r^+}}\pi_s\delta_s^2\right]\nonumber\\
&\qquad+\sum_{s\in\mathcal{S}_{N_r^+}}\pi_so(\delta_s^2)+\sum_{s\in\mathcal{S}_{N_r^-}}\pi_so(\delta_s^2).
\label{eq:time_average_utility_simp}
\end{align} 
In the following, we study the \textcolor{black}{behavior of each of the terms in RHS of \eqref{eq:time_average_utility_simp}.} The first term in \eqref{eq:time_average_utility_simp} can we rewritten as 
\begin{align}
&\frac{\mathcal{R}(N_r\mu_t)}{N_r\left(1+\frac{\delta_r^-}{N_r}\right)\left(1-\frac{\delta_r^+}{N_r}\right)}\left[1+\frac{\delta_r^-}{N_r}\sum_{s\in\mathcal{S}_{N_r^+}}\pi_s-\frac{\delta_r^+}{N_r}\sum_{s\in\mathcal{S}_{N_r^-}}\pi_s\right]\nonumber\\
&\quad=\frac{\mathcal{R}(N_r\mu_t)}{N_r\left(1+\frac{\delta_r^-}{N_r}\right)\left(1-\frac{\delta_r^+}{N_r}\right)}\left[1+\frac{\delta_r^-\pi_r^+}{N_r}-\frac{\delta_r^+\pi_r^-}{N_r}\right],
\label{eq:first_term}
\end{align}
where $\pi_r^+\triangleq\sum_{s\in\mathcal{S}_{N_r^+}}\pi_s$ and $\pi_r^-\triangleq\sum_{s\in\mathcal{S}_{N_r^-}}\pi_s$ denote the stationary probability of being in a state such that the receiver turns on after $N_r^+$ and $N_r^-$ slots, respectively. Next, using the energy conservation principle at the receiver
\begin{equation}
\pi_r^+\left(\frac{R}{N_r-\delta_r^+}\right)+(\pi_r^--p_{d}^r)\left(\frac{R}{N_r+\delta_r^-}\right)=\mu_r(1-p_o^r),
\label{eq:energy_consv}
\end{equation}
where $p_d^r$ and $p_o^r$ denote the probability of battery discharge and overflow, respectively. Simplifying the above equation, and using the result for discharge and overflow probability, $p_d^r$ and $p_o^r$, \textcolor{black}{derived at the start of this section}, we get
\be
\pi_r^+\delta_r^--\pi_r^-\delta_r^+=O(\delta_r^-)
\label{eq:Pi_delta}
\ee

On the other hand, the second term in \eqref{eq:time_average_utility_simp} is written as 
\begin{align}
&\mathcal{R}^{(1)}(N_r\mu_t)\left[\frac{1}{N_r^-}\sum_{s\in\mathcal{S}_{N_r^-}}\pi_s\left(\mu_t\delta_r^--N_r^-\delta_t^-+2k_s\delta_t\right)\right.\nonumber\\
&\left.\qquad\qquad\qquad+\frac{1}{N_r^+}\sum_{s\in\mathcal{S}_{N_r^+}}\pi_s\left(-N_r^+\delta_t^--\mu_t\delta_r^++2\ell_s\delta_t\right)\right]\nonumber\\
&=\mathcal{R}^{(1)}(N_r\mu_t)\left(-\delta_t^-+\frac{\mu_t\delta_r^-\pi_r^-}{N_r^-}+\frac{\mu_t\delta_r^+\pi_r^+}{N_r^+}\right.\nonumber\\
&\hspace{80pt}\left.+\frac{2\delta_t}{N_r^-}\sum_{s\in\mathcal{S}_{N_r^-}}\pi_sk_s+\frac{2\delta_t}{N_r^+}\sum_{s\in\mathcal{S}_{N_r^+}}\pi_s\ell_s\right)\nonumber\\
&=\mathcal{R}^{(1)}(N_r\mu_t)\left(-\delta_t^--\frac{\mu_t\delta_r^+\delta_r^-}{(N_r+\delta_r^-)(N_r-\delta_r^+)}\right.\nonumber\\
&\left.\qquad\qquad+\frac{\mu_tN_r}{(N_r+\delta_r^-)(N_r-\delta_r^+)}[\pi_r^-\delta_r^--\pi_r^+\delta_r^+]+A\right),
\label{eq:second_last_eqn}
\end{align}
where $A\triangleq \frac{2\delta_t}{N_r^-}\sum_{s\in\mathcal{S}_{N_r^-}}\pi_sk_s+\frac{2\delta_t}{N_r^+}\sum_{s\in\mathcal{S}_{N_r^+}}\pi_s\ell_s$. 
Using \eqref{eq:Pi_delta}, we have
\be
\pi_r^-\delta_r^--\pi_r^+\delta_r^+=\delta_r^--\delta_r^++(\pi_r^-\delta_r^+-\pi_r^+\delta_r^-).
\label{eq:final_eq}
\ee

Note that the quantity in \eqref{eq:second_last_eqn} converges to zero as $O(\delta_r^+)+O(\delta_r^-)+O(\delta_t^+)$, and $\sum_{s\in\mathcal{S}_{N_r^+}}\pi_s o(\delta_s^2)$, $\sum_{s\in\mathcal{S}_{N_r^-}}\pi_so(\delta_s^2)$ and the last but one term in \eqref{eq:time_average_utility_simp}, goes to zero as $O({\delta_r^+}^2)+O({\delta_t^-}^2)$ and $O({\delta_r^-}^2)+O({\delta_t^-}^2)$, respectively. The proof completes by noting that the right-hand side in \eqref{eq:first_term} converges to $\frac{\mathcal{R}(\mu_tN_r)}{N_r}$, and $N_r=\frac{N}{\lfloor \frac{N\mu_r}{ R}\rfloor}$.
\end{proof}
\subsection{Proof of Lemma~\ref{lem:diif_policy_uc_policy_c}}
\label{app:proof_lemma_diff_policy_uc_vs_c}
\begin{proof}
	\textcolor{black}{To prove the result, we consider a Markov chain $\mathcal{M}'$ which has the same state space  as the Markov chain $\mathcal{M}$, described in the proof of Lemma~\ref{lem:policy_perf_guarantee_uncons}, and its transition probabilities are governed by the policy and the harvesting statistics at both the nodes. The average throughput achieved by the policy ${P}^{uc}$ can be written as $$\mathcal{T}^{uc}=\sum_{s\in\mathcal{S}}\pi_s^{uc}\frac{\mathcal{R}_s^{uc}}{N_s},$$ where $\pi_s^{uc}$ denote the stationary probability of Markov chain $\mathcal{M}'$ being in the state $s\in\mathcal{S}$. Thus, the difference between the time-average throughput achieved by two policies, using \eqref{eq:time_average_utility}, can be written as
		\begin{align*}
		\mathcal{T}^c-\mathcal{T}^{uc}&=\sum_{s\in\mathcal{S}}\pi_s\frac{\mathcal{R}_s}{N_s}-\sum_{s\in\mathcal{S}}\pi_s^{uc}\frac{\mathcal{R}_s^{uc}}{N_s},\\
		&\overset{(a)}{=}\sum_{s\in\mathcal{S}}\frac{\mathcal{R}_s}{N_s}\left(\pi_s-\pi_s^{uc}\right),\\
		&\overset{(b)}{<}\frac{\mathcal{R}_{\max}^{N_r^-}}{N_r^+}\left(\pi_0^{uc}-\pi_0\right) <\frac{\mathcal{R}_{\max}^{N_r^-}}{N_r^+}\pi_0^{uc},
		\end{align*}  
		where $(a)$ follows from the fact the rate obtained in state $s$ is the same for both $\mathcal{P}^c$ and $\mathcal{P}^{uc}$, and (b) uses the fact that the stationary distribution sums to one, and the maximum achieved rate is $\mathcal{R}_{\max}^{N_r^+}$. This completes the proof. }
\end{proof}
\bibliographystyle{IEEEtran}
{\bibliography{IEEEabrv,bibJournalList,references}}

\end{document}